\documentclass[twocolumn]{article}

\usepackage[top=0.75in,left=0.75in,right=0.75in,bottom=1in]{geometry}
\usepackage{amsmath}
\usepackage{amssymb}
\usepackage{amsthm}
\usepackage{epigraph}
\usepackage[utf8]{inputenc}
\usepackage{afterpage}
\usepackage{tikz}
\usepackage{soul}
\usepackage[framemethod=tikz]{mdframed}
\usetikzlibrary{arrows,chains,decorations.pathreplacing}

\newcommand{\beq}{\begin{eqnarray}}
\newcommand{\eeq}{\end{eqnarray}}
\newcommand{\bhlpar}{\begin{mdframed}[innerleftmargin=0,innerrightmargin=0,hidealllines=true,backgroundcolor=yellow!20]}
\newcommand{\ehlpar}{\end{mdframed}}
\DeclareMathOperator{\wbf}{wBF}
\newtheorem{theorem}{Theorem}

\begin{document}

\title{Towards Attack-Tolerant Networks: Concurrent Multipath Routing and the Butterfly Network}
\author{
Edward L. Platt \\
University of Michigan\\
elplatt@umich.edu
\and
Daniel M. Romero \\
University of Michigan\\
drom@umich.edu
}
\date{}
\maketitle
\begin{abstract}
Targeted attacks against network infrastructure are notoriously difficult
to guard against.
In the case of communication networks,
such attacks can leave users vulnerable to
censorship and surveillance,
even when cryptography is used.
Much of the existing work on network fault-tolerance focuses on random faults and
does not apply to adversarial faults (attacks).
Centralized networks have single points of failure by definition,
leading to a growing popularity in decentralized architectures and protocols
for greater fault-tolerance.
However, centralized network {\em structure} can arise even when {\em protocols} are
decentralized.
Despite their decentralized protocols,
the Internet and World-Wide Web have
been shown both theoretically and historically to be highly susceptible to
attack, in part due to emergent structural centralization.
When single points of failure exist,
they are potentially vulnerable to non-technological
(i.e., coercive) attacks,
suggesting the importance of a structural approach to attack-tolerance.
We show how the assumption of partial trust transitivity,
while more realistic than the assumption underlying webs of trust,
can be used to quantify the effective redundancy of a network
as a function of trust transitivity.
We also prove that the effective redundancy of the wrap-around butterfly topology
increases exponentially with trust transitivity and
describe a novel concurrent multipath routing algorithm for constructing
paths to utilize that redundancy.
When portions of network structure can be dictated
our results can be used to create scalable,
attack-tolerant infrastructures.
More generally, our results provide a theoretical formalism for evaluating the effects
of network structure on adversarial fault-tolerance.
\end{abstract}

\noindent{\bf Keywords:}
Butterfly Topology,
Fault Tolerance,
Adversarial Faults,
Multipath Routing,
Censorship,
Decentralization.

\section{Introduction}

\epigraph{The Net interprets censorship as damage and routes around it.}{\textit{--John Gilmore} \cite{elmer-dewitt_first_1993}}

Much of the world's infrastructure is networked: power grids,
cellular networks, roads, and of course, the Internet.
As with any critical infrastructure, the cost of failures can be
immense, so methods for tolerating various kind of faults within networks are an
important and ongoing area of research
\cite{zin_survey_2015,albert_error_2000,sterbenz_resilience_2010}.
{\em Adversarial faults},
those in which an adversary can target attacks strategically,
deserve special attention.
Such attacks are both extremely difficult to guard against and 
often have important social implications.
In particular, censorship and surveillance are often achieved
by targeting central network locations and either blocking or capturing
the information flowing through them.
The Internet's decentralized design was motivated
by the need to withstand targeted attacks, such as nuclear strikes
\cite{baran_distributed_1964}.
But despite longstanding common wisdom \cite{elmer-dewitt_first_1993},
both theoretical results and recent events (described below)
have demonstrated that the Internet can be surprisingly vulnerable to attack.
Decentralization remains a promising approach towards
building resilient networks,
but there is a need to better understand the relationship between
decentralized network structure and adversarial fault tolerance.

The Internet's vulnerability to censorship and other targeted attacks
has been demonstrated by several recent events.
In 2008, YouTube suffered a worldwide outage for several hours
when a service provider in Pakistan advertised false routing information
\cite{hunter_pakistan_2008}.
The action (known as a {\em black hole attack}) was intended to censor YouTube
within Pakistan only, but resulted in a worldwide cascading failure.
The action was initiated by government order in Pakistan,
and spread beyond Pakistan when a router misconfiguration allowed the false
routing information to propagate globally.
While the government order and router misconfiguration initiated the outage,
it was the structure of the Internet's router network that allowed a fault in a
single router to propagate globally.
And while the action was not an intentional attack against the global Internet,
the ability of an attacker to succeed without even trying only highlights
the Internet's vulnerability to adversarial faults.

In 2013, the Texas-based email provider Lavabit was ordered to disclose
their private SSL keys to the FBI \cite{poulsen_edward_2013}.
Rather than complying,
Lavabit ceased operations
in order to protect their users from surveillance.
Once again, the attack was successful due to a highly centralized
architecture:
SSL keys under control of a single entity, in a single legal jurisdiction.
While originally intended as surveillance,
this action effectively became an act of censorship.
It is also important to note that while Lavabit's cryptography worked as intended,
the attack was still successful because the system was
vulnerable to non-technical coercion.
So we see that such vulnerabilities are not limited to any one system or protocol,
but result from centralized structure itself.

Analysis of the Internet's router network has shown that while it
is remarkably resilient against random faults,
it is highly susceptible to adversarial faults \cite{albert_error_2000}.
These results have been attributed to the scale-free structure of the Internet's
router network
\cite{barabasi_emergence_1999,barabasi_scale-free_2009}.
In scale-free networks and other networks with heavy-tail degree distributions,
random failures are highly likely to
affect only low-degree nodes, thus having
little effect.
However, Adversarial faults target the few high-degree nodes,
and therefore remove a large number of edges with each fault.
So while the {\em protocols} of the Internet are decentralized,
the {\em network structure} is somewhat centralized. 
In other words, the protocols of the Internet do not {\em require}
centralization, but centralization may still emerge from the sociotechnical
processes that create its network structure.

With strong theoretical and historical evidence that centralized
network structure can create vulnerabilities,
methods for analyzing structural vulnerabilities and for designing
fault-tolerant networks are needed.
This paper presents several contributions towards advancing those goals.
While our motivation comes from the Internet router network and
World Wide Web,
our work is theoretical, focusing on abstract networks,
and could potentially be applied to many different types of networks,
whether made of physical wires, virtual tunnels,
or other types of links.

We consider a setting in which a source node (Alice)
in a network attempts to route a message
to a target node (Bob) by forwarding it through the links of the network.
A ``link'' in this context could represent any kind of connection
(e.g., physical cables, encrypted channels).
We assume that some nodes in the network may be compromised by an attacker
(Mal).
We also assume that Mal is an adversary of Alice specifically and targets
nodes strategically with the goal of interfering with Alice's communications
(rather than disrupting the network as a whole).
This assumption applies to scalable network architectures that can be made
large enough that an attacker must focus their resources in order to have
a significant effect.
Compromised nodes may behave incorrectly by blocking, altering,
or incorrectly routing messages.
We assume that Mal has full knowledge of the network structure, but has
limited resources and thus can only compromise a fixed number of nodes.

We also assume that nodes {\em trust} their immediate neighbors.
So Mal is unable to compromise Alice's node, or her direct neighbors.
In the commonly used {\em web of trust} approach
\cite{zimmermann_official_1995,ferguson_practical_2003},
we would extend that trust transitively to the entire network.
However, we make a weaker and more realistic assumption:
that trust is extended transitively to nodes within a fixed
number of hops.
So only the beginning and end of a path between Alice and Bob is trusted
against interference from Mal.
We call this assumption
{\em partial trust transitivity},
and refer to such paths as {\em partially trusted}.

Under the above assumptions,
we show how to evaluate the influence of network structure
on attack-tolerance,
how to use local trust and redundancy to achieve greater attack-tolerance
when no single path is fully trusted,
and propose a novel routing algorithm for constructing such paths on
the butterfly network topology.
The butterfly topology is popular in parallel processing
\cite{kshemkalyani_distributed_2008} and
peer-to-peer \cite{lua_survey_2005, korzun_structured_2013}
applications, due to its regular structure, low degree, and high connectivity.

It is important to note that the butterfly is a highly structured and constrained
network topology,
very different form those found in social networks and other
self-organized networks.
The reader may wonder whether it is realistic or useful to assume such control over
the network structure.
We have already seen that whenever a single point of failure exists in the network,
there is a potential for an attacker to exploit it through coercion,
without needing to compromise the technology.
So, {\em attack-tolerance cannot be achieved without the ability to influence
network structure}.
Luckily, there are scenarios in which network topology can be dictated.
Examples include overlay networks
\cite{lua_survey_2005, korzun_structured_2013},
formal organizations \cite{mohr_explaining_1982},
government-regulated cellular networks \cite{walker_mass_2012},
and call tree notification systems \cite{nickerson_thinking_2010}.
In general,
{\em when the need for attack-tolerance is high enough to warrant investment
in infrastructure, networks can be engineered and maintained as infrastructure}.
It is also worth noting that attack-tolerant networks may be sub-components of
larger, less-constrained systems.
For example, a single server might be replaced by a distributed network of servers,
each with different ownership, physical location, and legal jurisdiction,
without placing any unrealistic constraints on the clients connecting to
those servers.
Additionally, there may be ways to achieve improved attack-tolerance from architectures
more flexible than the butterfly, which is a potential area for future work.

We begin by describing how fault tolerance techniques can be adapted and
evaluated in a network setting with partial trust transitivity and
adversarial faults.
Generally, faults in network paths can be correlated,
preventing the application of standard fault tolerance techniques
\cite{avizienis_basic_2004, von_neumann_probabilistic_1956},
which assume independent faults.
By constructing {\em independent paths},
which have no untrusted nodes in common,
we show how to model communication across a complex network in the presence
of correlated adversarial faults as communication across redundant
simple channels with random errors.
Redundant messages can be sent across these channels in parallel,
a technique known as {\em concurrent multipath routing}
\cite{zin_survey_2015, qadir_exploiting_2015, khiani_comparative_2013},
and used for fault tolerance.
The receiver can then use the redundant messages to detect and/or correct
errors.
We formally evaluate the
effects of network structure on attack-tolerance and show that the probability
of an undetected error decreases exponentially with the number of
independent paths between source and destination,
even when no individual path is entirely trusted.

We also propose a novel concurrent multipath routing algorithm for the butterfly
topology.
The algorithm constructs independent paths,
which can be combined with the fault-tolerant concurrent multipath routing
scheme above to
achieve a high level of adversarial fault tolerance on the butterfly topology.

Our main contributions are:
\begin{itemize}
\item{
We propose a novel method for extending standard fault tolerance techniques to
{\em adversarial} faults in {\em complex networks}.
We do so by modeling redundant independent paths
with partial trust transitivity as a single virtual channel,
and show that the probability of detecting adversarial faults
approaches 1 exponentially with the number of paths;
}
\item{We prove that the number of independent paths between two nodes
in a wrap-around butterfly network with partial trust transitivity
increases exponentially with the trust radius;
}
\item{We present a scalable, efficient, and attack-tolerant concurrent
multipath routing algorithm on the butterfly network topology.}
\end{itemize}

This paper is organized as follows.
Section \ref{sec-related} reviews background and related work.
Section \ref{sec-ft} describes adversarial fault tolerance on
structured networks.
Section \ref{sec-butterfly} describes the concurrent multipath routing
algorithm for the butterfly network topology
Section \ref{sec-discussion} discusses the results.
And Section \ref{sec-conclusion} concludes.

\section{Background and Related Work}
\label{sec-related}

There has been considerable work on trust in network security.
Both centralized and decentralized approaches are commonly used to create
trust infrastructures.
Centralized approaches such as {\em public key infrastructure} (PKI)
suffer from a number of vulnerabilities
\cite{ellison_ten_2000},
which stem largely from the single points of failure inherent to
centralization.
The well-known and widely-used {\em web of trust} approach
\cite{zimmermann_official_1995,ferguson_practical_2003}
is a decentralized alternative.
In a web of trust,
individuals choose who they trust initially.
Trust is then extended to new individuals if they are vouched for by a
currently-trusted individual,
making it possible to quickly establish a large group of trusted nodes.
However, web of trust's assumption of infinite transitivity is unrealistic
\cite{christianson_why_1997},
and does not distinguish between paths of different lengths.
Our work addresses both of these limitations by incorporating a more realistic
assumption of partial transitivity.

Previous work on incorporating network structure into
trust models has focused on authentication
protocols, showing that independent paths can reduce an adversary's ability
to impersonate a target
\cite{levien_attack-resistant_2009}.
Other work has shown that identifying independent paths in arbitrary networks
is NP-hard and provided approximation algorithms
\cite{reiter_resilient_1998}.
Our work complements these by introducing the partial trust assumption
extending the focus beyond authentication.
When network topology can be controlled, we sidestep the NP-hard problem of finding
independent paths on arbitrary networks by using the mathematical structure of
the butterfly topology to construct provably independent paths.

Many distributed consensus protocols (such as those used by cryptocurrencies)
are designed to tolerate arbitrary or adversarial faults.
Byzantine agreement protocols
\cite{lamport_byzantine_1982,castro_practical_1999}
provide tolerance against arbitrary faults (including attacks) under
some circumstances, but are limited to small networks due to poor scalability.
Proof-of-work \cite{dwork_pricing_1993,nakamoto_bitcoin:_2008}
and proof-of-stake \cite{king_ppcoin:_2012}
provide better scalability,
but are wasteful of computational and energy resources.
Federated Byzantine Agreement (FBA) \cite{mazieres_stellar_2015}
is scalable, allows for flexible trust,
and is highly fault-tolerant on networks meeting a set of requirements.
However, FBA does not provide a method for evaluating the
fault tolerance properties of different network structures
or for calculating the failure probabilities within a particular network.

There are relatively few attack-tolerance schemes
that focus on network structure,
compared to more popular cryptographic approaches
\cite{ferguson_practical_2003}.
All existing attack-tolerant networks we are aware of are content-addressable
networks: data is routed to and from storage nodes rather than between sender
and receiver.
Fiat and Saia described a scheme that combines the butterfly topology
with expander graphs to create a highly censorship-resistant,
content-addressable network \cite{fiat_censorship_2002},
although this scheme does not scale well and is impractical due to a
high level of data replication.
Perhaps the most mature structural solution is the Freenet collaboration
\cite{clarke_freenet:_2001}.
Freenet uses secret sharing
\cite{shamir_how_1979, blakley_safeguarding_1979}
and small-world routing
\cite{zhang_using_2002,kleinberg_small-world_2000}
to create a content-addressable network with a high level of both
confidentiality and censorship resistance.
Freenet guarantees that data is stored redundantly,
but still allows for centralized network structure,
and thus single points of failure,
as data travels from its origin to the redundant storage locations.
Unlike the above content-addressable networks, our proposal is purely network based
and does not require nodes to store data indefinitely.
Our proposal also improves on the scalablity of Fiat and Saia's work,
and does not rely on assumptions about existing social network structure.

{\em Multipath routing} protocols identify multiple paths between
source and destination
in contrast to traditional {\em unipath} routing, which uses
a single path.
The special case of {\em concurrent} multipath routing uses mutliple paths
simultaneously.
Multipath routing has many applications, including reduced congestion,
increased throughput, and more reliability
\cite{qadir_exploiting_2015}.
Many of these routing protocols offer increased confidentiality
\cite{zin_survey_2015}.
Some approaches utilize redundant paths as backups for increased
fault tolerance
\cite{alrajeh_secure_2013},
and some specifically protect against adversarial faults
\cite{kohno_improvement_2012, khalil_unmask:_2010, lou_h-spread:_2006}.
Most work on multipath routing has been motivated by applications related to
wireless sensor networks (WSNs),
and have thus focused on ad hoc, unstructured networks, often having a central
base station.
The method of Liu et al.
\cite{liu_secure_2012}
routes multiple messages first to random peers and then
to a central base station,
with the network edges constrained by sensors' physical location.
We have found only few examples in the existing literature of applications of
concurrent multipath routing to {\em adversarial} fault tolerance,
and all have focused on ad-hoc wireless sensor networks, without attention
to the role of network structure.
The alogorithm we present for the butterfly topology complements existing work
by addressing cases where links are not constrained by physical
distance,
and where network structure can be engineered for greater attack-tolerance.

Our proposed routing algorithm makes use of a
{\em structured network}, in which link structure is predetermined.
Structured networks have been a popular tool in parallel processing
architectures \cite{kshemkalyani_distributed_2008}.
More recently, peer-to-peer systems based on distributed hash tables have used
structured {\em overlay networks} to map table keys to local TCP/IP routes
\cite{lua_survey_2005,korzun_structured_2013}.
Such networks can be designed to have favorable structural and routing
properties,
which can be used to to improve attack-tolerance.

\section{Trust Networks and Fault Tolerance}
\label{sec-ft}

Within the field of {\em fault tolerance},
many techniques have been developed for building reliable systems
out of unreliable components
\cite{avizienis_basic_2004, von_neumann_probabilistic_1956}.
We will make use of standard fault tolerance terminology, summarized here.
A {\em fault} is said to occur when one component
of a system behaves incorrectly (e.g., a routing node blocks or
alters a message).
The result of that fault (e.g., a recipient receiving conflicting messages)
is called an {\em error} state.
If the error is undetected or corrected to the wrong value, the system is
said to have experienced a {\em failure} (e.g., an altered message is
accepted as authentic).
Note that when an error is detected but cannot be corrected,
the system has still tolerated the fault because it has not accepted an error
state.
We are concerned in particular with {\em adversarial faults},
which are chosen strategically to maximize the likelihood of a failure.

\subsection{Partial Trust Model}

A central question in large-scale, secure communication is this:
how can two parties communicate reliably and securely
when no direct trusted link exists between them?
The commonly-used web of trust approach
\cite{zimmermann_official_1995,ferguson_practical_2003}
extends trust infinitely transitively:
to friends of friends, and friends of friends of friends, and so on.
However, the assumption of infinitely transitive trust is unrealistic
\cite{christianson_why_1997},
and does not allow for the analysis of the effects of network structure.

An alternative assumption might be that each hop away from Alice in
in the network
reduces the probability that a node can resist compromise exponentially.
Such a situation could occur if nodes more distant from Alice are
more favorably disposed to Alice's adversary, more likely to cooperate with that
adversary, or less likely to take proactive security measures against that
adversary.
The above model can be further
approximated by assuming that nodes up to some fixed number
of hops cannot be compromised, and that those beyond can.
This simplified version is still more realistic than infinite transitivity
and will be convenient for proving our results.
We now proceed to define our model formally.

We define the {\em partial trust model} on
an undirected graph $G = (V,E)$,
although the model can easily be extended to directed multigraphs.
Vertices representing commiunicating agents,
and with edges representing mutually trusted communication links.
Let $v \in V$ be an arbitrary sender (Alice)
and $w \in V$ be an arbitrary receiver (Bob).
We assume the presence of an adversary (Mal) who knows the
full structure of the network,
and who can compromise a fixed number of nodes,
gaining complete control of their behavior.
We also assume that Mal is an adversary of Alice and/or Bob specifically,
rather than the network as a whole.
So Mal can compromise any node except for those trusted by Alice or Bob.
We define a {\em trust radius} $h$ such that nodes $u$ and
$u^\prime$ trust each other if their distance is less than $h$.
For a given node $u$,
we call the set of trusted nodes its
{\em trusted neighborhood} $T_h(u)$,
and all nodes at exactly distance $h$ the
{\em trust boundary} $B_h(u)$:
\beq
T_h(u) &=& \left\{ u^\prime \mid d(u,u^\prime) < h \right\} \\
B_h(u) &=& \left\{ u^\prime \mid d(u,u^\prime) = h \right\}.
\eeq
The trust boundary $B_h$ plays an important role because these nodes are not
trusted by $u$,
and if compromised can entirely isolate $u$ from the rest of the network.
These trust assumptions imply that when Alice sends a message to Bob,
Mal can only cause faults in the set of nodes outside both of their trusted
neighborhoods: $V \setminus \left(T_h(v) \cup T_h(w)\right)$.
We refer to this set of nodes as the {\em untrusted region}.

\subsection{Effective Redundancy}

Our goal is to achieve fault tolerance through redundancy.
In the network setting, redundnacy is achieved using
{\em independent paths} \cite{reiter_resilient_1998},
which have no common points of failure.
Typically, it is assumed that paths must be disjoint in order to be
independent.
However, under the partial trust assumption, two non-disjoint paths
can still be independent as long as their intersection contains only trusted
nodes,
greatly increasing the level of redundancy available.
Under the partial trust assumption, the available redundancy thus depends on
both the network structure and the level of trust.

We now quantify the {\em effective redundancy} between Alice and Bob
when trust radius $h$ is assumed.
This quantity, $\delta_{v,w,h}$ is exactly the max-flow/min-cut of
the graph after Alice's and Bob's trusted neighborhoods have been
collapsed into single source/sink vertices.
Each trust boundary forms a cut of the network and places an upper bound on the
min-cut:
\beq
\delta_{v,w,h} \leq \min\left( \mid B_h(v) \mid, \mid B_h(w) \mid \right).
\eeq
Equality holds when there are no bottlenecks within the untrusted region,
an indication that the network is decentralized.
The efrective redundancy of the entire graph can be characterized by the minimum over
all vertex pairs:
\beq
\delta_h(G) \equiv \min_{v,w \in V} \delta_{v,w,h}.
\eeq
Thus, for any pair of nodes in the network, at least $\delta_h$ independent,
redundant paths can be constructed between them.
$\delta_h$ is a purely structural network property,
and places an upper bound on the effectiveness of any
redundancy-based fault tolerance scheme.
The more quickly $\delta_h$ grows with $h$,
the better a network is at leveraging trust transitivity to create redundancy.
Thus, the scaling of $\delta_h$ can be used to quantify a network's ability
to withstand targeted attacks,
even when the exact trust radius $h$ is unknown.

\begin{figure}
\centerline{\includegraphics[width=3in,height=2.08in]{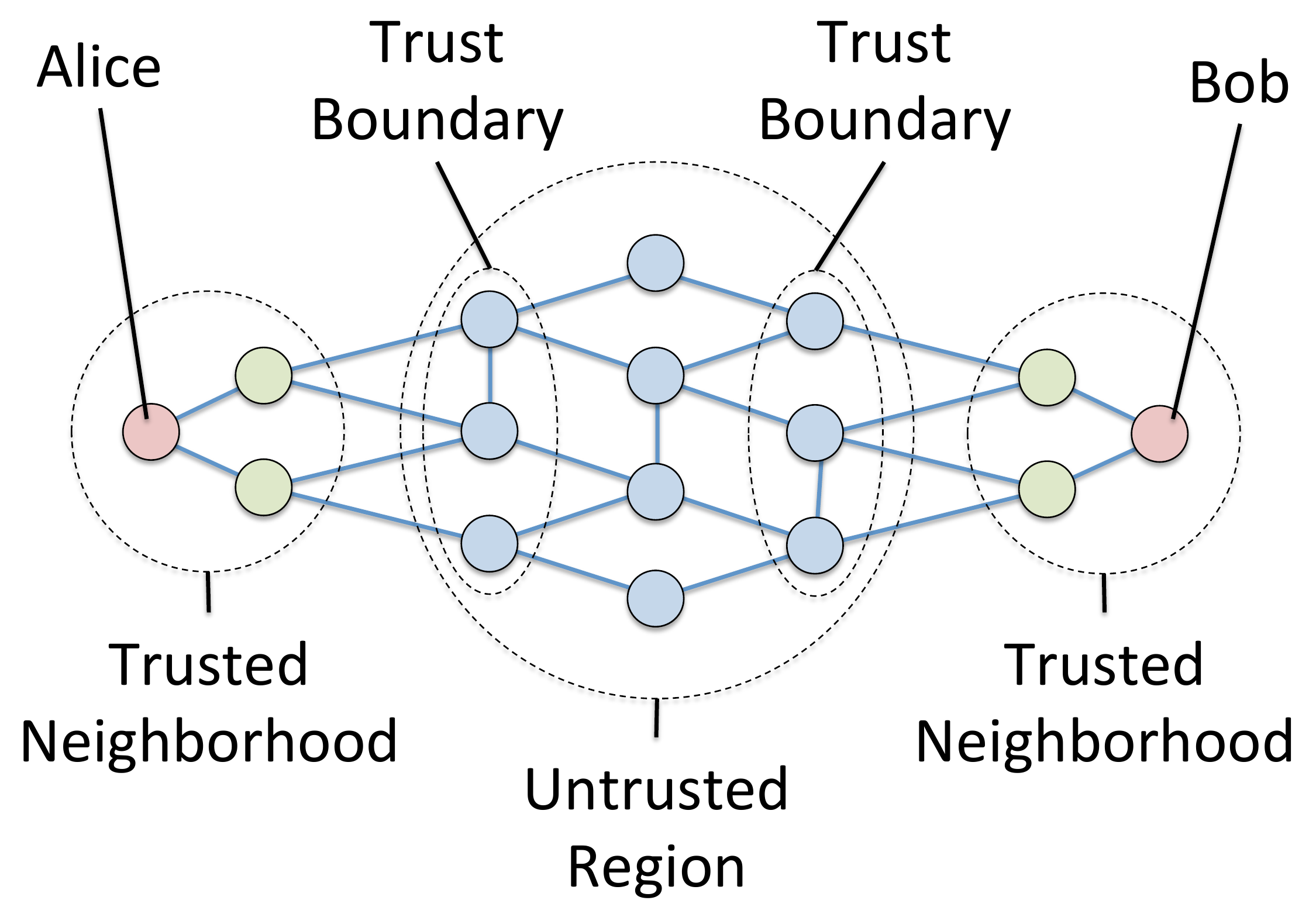}}
\caption{
Illustration of a trusted communication network and the network properties
used by the {\em partial trust model}.
Edges represent mutually trusted communication links.
The sender (Alice, $v$) and receiver (Bob, $w$) trust all nodes
less than the {\em trust radius} $h$ hops away.
These nodes form their {\em trusted neighborhoods} $T_h(v)$ and $T_h(w)$.
We assume that all faults occur in the remaining nodes: the
{\em untrusted region}.
The untrusted nodes in contact with the trusted neighborhoods for the
{\em trust boundaries} $B_h(v)$ and $B_h(w)$,
which (in the absence of central bottlenecks) determine the
{\em effective redundancy} $\delta_h$ provided by the network.
Alice and Bob can be modelled as connected by a direct link with
at least $\delta_h$ redundant channels.
}
\label{fig:trust-source}
\end{figure}

\subsection{Multipath Fault Tolerance}

Once we have determined a network's effective redundancy,
we can apply redundancy-based fault tolerance techniques,
by sending multiple copies of a message
({\em concurrent multipath routing}).
Having found the effective redundancy between two nodes,
we can simplify our model, replacing independent paths through the
complex network with direct channels between the endpoints.
We model our sender (Alice) and receiver (Bob) as
communicating over $\delta_h$ direct and redundant virtual channels.
The partial trust model allows us to make this simplifying assumption
for analyzing a fault tolerance scheme,
but implementing such a scheme will require a method for constructing
specific network paths.
We will return to the question of constructing paths in the next section.
For now, we concern ourselves with the question:
given that the network provides $\delta_h$ redundant channels between
Alice and Bob,
what is the probability that an adversary (Mal) causes an undetectable
error after inducing a fixed number of faults?

Let us first consider the scenario in which
Alice sends a message copy over each available channel.
We can also assume that each message includes the number of messages sent,
the full list of channels used, etc.,
making that information available to Bob.
When Bob receives the messages, there are several possibilities.
If some of the messages are missing
or if some of the messages disagree,
Bob knows that some of the messages were either blocked or altered,
and he has successfully tolerated the fault(s).
Bob can then take any of several actions:
1. request retransmission;
2. end receipts so Alice knows which paths have been compromised;
or 3. attempt error correction using majority voting.
If instead, Bob finds that all the messages are present and agree,
there are two possible cases.
The first case is that Mal has not compromised any of the messages,
and Bob has correctly accepted them, so no failure has occurred.
The second case is that Mal has compromised {\em all} of the messages,
so Bob has accepted an erroneous message and a failure has occurred.
In the present scenario,
whether a failure occurs depends only on whether Mal has the resources to
compromise all of the channels.
In a more realistic scenario,
both Alice and Mal have limited resources and are not able to use or
compromise all available channels.

In a more sophisticated multipath fault tolerance scheme,
Alice randomly chooses $k \leq \delta_h$ channels and sends a copy of
her message on each.
We assume that Mal is capable of compromising $l \leq \delta_h$ channels.
Since Alice chooses channels randomly,
all channels are equally likely to contain a message,
so Mal can do no better than also choosing randomly.
We can also return to the full network setting by noting that each of the
$\delta_h$ independent paths in the network can serve as independent channels
between Alice and Bob.
Mal's best strategy is now to identify a minimum node cut in the network
and randomly compromise nodes from that cut.
With this strategy, each compromised node reduces effective redundancy by one,
equivalent to compromising one of the channels between Alice and Bob.
If $k > l$, at least one message will get through uncompromised and all
errors are detectable.
Otherwise, the probability of Mal producing an undetectable error is
the probability that all of Alice's chosen channels are compromised:
\beq
\label{eq:pf}
p_f &=& \frac{l!(\delta_h-k)!}{\delta_h!(l-k)!}.
\eeq
Letting $k=\alpha \delta_h$ and $l=\beta \delta_h$, then applying Stirling's
approximation gives:
\begin{eqnarray}
\label{eq:pf_approx}
p_f &\approx&
\frac{\sqrt{\beta(1-\alpha)}}{\sqrt{\beta-\alpha}}
\left[
    \left( \frac{\beta-\alpha}{1-\alpha} \right)^{\alpha}
    \left( \frac{\beta}{\beta-\alpha} \right)^{\beta}
    (1-\alpha)
\right]^{\delta_h}.
\end{eqnarray}

Figure \ref{fig:pfail} shows the value of $p_f$
as a function of $k$ and $l$.
Equation (\ref{eq:pf_approx}) shows that while $p_f$
depends on the fractions of
redundant paths actually utilized $\alpha$ and compromised $\beta$,
it decreases exponentially with the effective redundancy $\delta_h$
(which we will later see increases exponentially in $h$ in the
butterfly topology).
This result is significant because $\delta_h$
depends only on the network structure
and the strength of trust transitivity.
{\em Thus, the scheme can be effective, even when the number of copies
sent $k$ is a small fraction of the effective redundancy}.
In other words, this scheme exhibits a {\em stabilizing asymmetry}:
senders can tolerate attacks from significantly more powerful
adversaries,
as long as the network structure provides large $\delta_h$.

In order to derive the above results,
we have assumed that Alice and all
intermediary agents are able to identify specific,
independent network paths that achieve the effective redundancy $\delta_h$.
We now proceed to describe a routing algorithm for doing so in the special case
of the butterfly network topology.

\begin{figure}
\centerline{\includegraphics{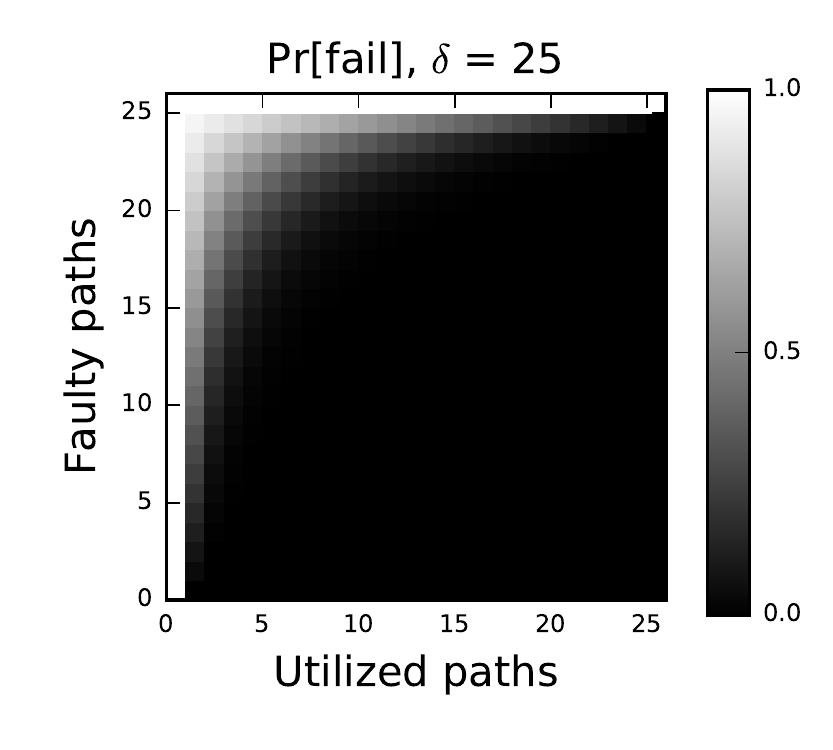}}
\caption{
The probability of an undetectable error as a function of the number of
message copies and the number of adversarial faults.
}
\label{fig:pfail}
\end{figure}

\section{Multipath Butterfly Routing}
\label{sec-butterfly}

In previous sections, we showed that reliable communication across a network
can be achieved even when any single message path might be compromised by
an adversary,
provided the network has sufficient redundancy,
and provided the sender and intermediaries know how to route message
copies along independent paths.
In this section, we address both requirements by proposing a novel routing
algorithm for constructing independent paths on the butterfly network topology.
This architecture and routing algorithm achieve an
effective redundancy that increases exponentially with the trust radius,
allowing a very high level of adversarial fault tolerance.

The structure of the butterfly network is highly constrained,
making it most suitable for applications where portions of the
network structure can be designed or dictated.
Examples of such networks include:
overlay networks \cite{lua_survey_2005, korzun_structured_2013},
formal organizations \cite{mohr_explaining_1982},
government-regulated cellular networks \cite{walker_mass_2012},
and call tree notification systems \cite{nickerson_thinking_2010}.
However, when attack-tolerance is desired,
it will always require control over network structure in order to
eliminate single points of failure.
The regular structure of the butterfly is not a limitation of our approach,
but rather a reflection of the inherent difficulty of attack-tolerance.
Lastly, we note that the partial trust model and multipath fault tolerance
schemes of the previous section do not rely on any particular network
topology or routing algorithm,
and our choice of the butterfly topology is only one of many possible choices.

\subsection{Butterfly Network Topology}

We choose the butterfly topology
\cite{kshemkalyani_distributed_2008}
because of several desirable properties (described below)
and because its structure allows for relatively straightforward
design and analysis of routing algorithms.
While several variations on the butterfly network exist,
we utilize the wrap-around butterfly.
We denote the $m$-dimensional, directed wrap-around butterfly as a
graph $\wbf(m)$:
\beq
\wbf(m) &=& (V, E_\downarrow \cup E_\rightarrow) \\
V &=& \mathbb{Z}_{m} \times \mathbb{Z}_2^m \\
E_\downarrow
&=&
\{((l,z),(l+1 \; (\text{mod } m),z) \} \\
E_\rightarrow
&=&
\{(l,z),(l+1 \; (\text{mod } m),
z \oplus 1_l \},
\eeq
where $\mathbb{Z}_m$ is the set of integers modulo $m$,
$\oplus$ represents componentwise addition modulo 2,
and $1_l$ is a vector with a 1 in index $l$ and 0 elsewhere.
Each node is associated with a level $l$ and an $m$-bit string $z$
known as {\em the place-within-level}.
There are two types of edges: down, and down-right
(shown in Figure \ref{fig:butterfly}).
Down edges ($E_\downarrow$) connect nodes sharing the same $z$ value
in a cycle of increasing level $l$.
Down-right edges ($E_\rightarrow$) also link to a node of level $l + 1$,
but one having the place-within-level equal to $z$ with the $l$th bit inverted.

The wrap-around butterfly network is known to have several of the properties
we desire for scalable, decentralized communication networks:
\begin{description}
\item[Vertex-transitivity:]
Because the wrap-around butterfly is vertex transitive,
it is maximally decentralized;
\item[Small-diameter:]
For any two nodes, the length of the shortest path between them is
$O(\log N)$, where N is the number of nodes in the network;
\item[Sparsity:]
With a constant degree of 4, the wrap-around butterfly is extremely sparse,
and can scale indefinitely without node degree becoming a limitation;
\item [Redundancy:]
Multiple paths exist between any two nodes.
Specifically, we will prove below that the number of independent paths between two
nodes increases exponentially with the trust
radius $h$.
\end{description}

\begin{figure}
\begin{center}
\includegraphics{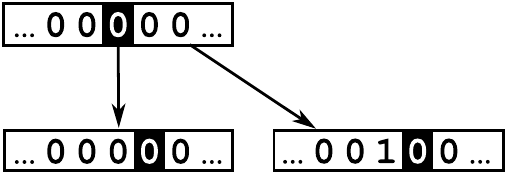}
\end{center}
\caption{
Schematic illustration of the two types of edges in a directed butterfly
network.
The node $(l,z)$ is shown as the bit string $z$ with a square around the
$l$th bit.
``Down'' edges increment $l$, leaving $z$ unchanged,
while ``down-right'' edges increment $l$ and invert the $l$th bit of $z$.
In the wrap-around variant, the nodes with maximum $l$ have down and down-right
edges to the nodes with $l=0$.
\label{fig:butterfly}
}
\end{figure}

The structure of the butterfly network lends itself to a well-known
(unipath) routing algorithm,
which we later extend to the multipath case.
The unipath algorithm first follows a down or down-right edge at every step,
increasing the level $l$ by 1 and cycling through the
indices of the place-within-level.
If the current node's place-within-level matches the destination node's at
index $l$,
a down edge is chosen and the place-within-level does not change.
Otherwise, a down-right edge is chosen and the $l$th component of the
place-within-level is flipped,
after which it matches the destination.
After $m$ iterations of this, all levels have been visited
and the place-within-level matches that of the destination.
Simply following down (or up) edges will then increment (decrement) the
level until the destination node is reached.

\subsection{Multipath Routing Algorithm}

We now present a routing algorithm to construct $2^h$ independent paths
between two nodes in a butterfly network,
where $h$ is the trust radius under the partial trust model.
Informally, Alice sends each message to a distinct node on her
trust boundary, then to a distinct intermediate node in the untrusted region,
then to a distinct node on Bob's trust boundary, and finally to Bob.
The intermediate nodes are in a sense ``far'' from each other and ensure that
no two paths overlap in the untrusted region.
Each path can be parameterized by a single integer $s$, which identifies
the specific node on Alice's trust boundary
(or equivalentely the node on Bob's trust boundary, or the untrusted intermdiate).

The algorithm guarantees paths are independent by ensuring that
(outside the trusted neighborhoods)
they only include
nodes that match the path parameter $s$ at certain indexes in their
place-within-level.
Since each path has a unique parameter $s$,
its set of untrusted nodes is disjoint from all other paths.
As with the unipath routing algorithm,
each of the multiple paths proceed from a source $v$ to a destination $u$
using down and down-right edges,
cycling through levels one at a time.
However, we cycle through the levels twice, once to route from $v$ to a
particular path's intermediary node,
and again to route from the intermediary to $w$.
Each cycle is divided into stages,
with different properties used to prove independence at each stage
(see Figure \ref{fig:route-overview}).
In the first cycle (stages 1--4), path independence is guaranteed by ensuring that
all nodes match the path parameter $s$ in the first $h$ bits of the place-within-level.
Similarly, in the second cycle (stages 5--7),
independence is guaranteed by ensuring that all paths match $s$ in the
$h$ bits of the place-within-level preceding the destination index.
A full example is illustrated in Figure \ref{fig:routing}.

\afterpage{
\begin{table}[h!]
\caption{Butterfly Multipath Routing Variables\label{tab:routing}}
{
\begin{center}
\begin{tabular}{|l|l|}
\hline
NAME & VARIABLE \\\hline
butterfly dimension & $m \in \mathbb{Z}_+$ \\\hline
node level & $l \in \mathbb{Z} : 0 \leq l < m$ \\\hline
node place within level & $z \in \mathbb{Z}_2^m$ \\\hline
trust radius & $h \in \mathbb{Z} : 1 \leq h \leq \lfloor m/2 \rfloor$ \\\hline
path index & $s \in \mathbb{Z}_2^h$ \\\hline
\end{tabular}
\end{center}
}
\end{table}

\begin{figure}[t!]
\begin{center}
\begin{tikzpicture}[
node distance=0pt,
 start chain = A going right,
    X/.style = {rectangle, draw,
                minimum width=10ex, minimum height=3ex,
                outer sep=0pt, on chain},
                        ]
\foreach \i in {0\ldots,{\ldots}0\ldots,{\ldots}0\ldots,{\ldots}0}
    \node[X] {\i};
\draw[<->] ([yshift=1.5mm] A-1.north east) -- node[above=0.25mm] {$h$} ([yshift=1.5mm] A-1.north west);
\draw[<->] ([yshift=1.5mm] A-2.north east) -- node[above=0.25mm] {$l_w - 2h$} ([yshift=1.5mm] A-2.north west);
\draw[<->] ([yshift=1.5mm] A-3.north east) -- node[above=0.25mm] {$h$} ([yshift=1.5mm] A-3.north west);
\draw[<->] ([yshift=1.5mm] A-4.north east) -- node[above=0.25mm] {$m - l_w$} ([yshift=1.5mm] A-4.north west);
\draw ( A-1.west) -- node[left=5ex,minimum width=10ex] {start} ( A-1.west);
\node (B1) [inner sep=1pt,above=of A-1.north,above=5ex] {\underline{A}};
\node (B2) [inner sep=1pt,above=of A-2.north,above=5ex] {\underline{B}};
\node (B3) [inner sep=1pt,above=of A-3.north,above=5ex] {\underline{C}};
\node (B4) [inner sep=1pt,above=of A-4.north,above=5ex] {\underline{D}};
\end{tikzpicture}
\\
\begin{tikzpicture}[
node distance=0pt,
 start chain = A going right,
    X/.style = {rectangle, draw,
                minimum width=10ex, minimum height=3ex,
                outer sep=0pt, on chain},
    Y/.style = {rectangle, draw,
                minimum width=10ex, minimum height=3ex,
                outer sep=0pt, on chain, thick},
                        ]
\node[Y] {$s$};
\foreach \i in {{\ldots}0\ldots,{\ldots}0\ldots,{\ldots}0}
    \node[X] {\i};
\draw ( A-1.west) -- node[left=5ex,minimum width=10ex] {1.} ( A-1.west);
\end{tikzpicture}
\\
\begin{tikzpicture}[
node distance=0pt,
 start chain = A going right,
    X/.style = {rectangle, draw,
                minimum width=10ex, minimum height=3ex,
                outer sep=0pt, on chain},
    Y/.style = {rectangle, draw,
                minimum width=10ex, minimum height=3ex,
                outer sep=0pt, on chain, thick},
                        ]
\foreach \i in {$s$}
    \node[X] {\i};
\node[Y] {{\ldots}1\ldots};
\foreach \i in {{\ldots}0\ldots,{\ldots}0}
    \node[X] {\i};
\draw ( A-1.west) -- node[left=5ex,minimum width=10ex] {2.} ( A-1.west);
\end{tikzpicture}
\\
\begin{tikzpicture}[
node distance=0pt,
 start chain = A going right,
    X/.style = {rectangle, draw,
                minimum width=10ex, minimum height=3ex,
                outer sep=0pt, on chain},
    Y/.style = {rectangle, draw,
                minimum width=10ex, minimum height=3ex,
                outer sep=0pt, on chain, thick},
                        ]
\foreach \i in {$s$,{\ldots}1\ldots}
    \node[X] {\i};
\node[Y] {$\tilde{s}$};
\foreach \i in {{\ldots}0}
    \node[X] {\i};
\draw ( A-1.west) -- node[left=5ex,minimum width=10ex] {3.} ( A-1.west);
\end{tikzpicture}
\\
\begin{tikzpicture}[
node distance=0pt,
 start chain = A going right,
    X/.style = {anchor=base, rectangle, draw,
                minimum width=10ex, minimum height=3ex,
                outer sep=0pt, on chain},
    Y/.style = {rectangle, draw,
                minimum width=10ex, minimum height=3ex,
                outer sep=0pt, on chain, thick},
                        ]
\foreach \i in {$s$,{\ldots}1\ldots,$\tilde{s}$}
    \node[X] {\i};
\node[Y] {$z_{w,D}$};
\draw ( A-1.west) -- node[left=5ex,minimum width=10ex] {4.} ( A-1.west);
\end{tikzpicture}
\\
\begin{tikzpicture}[
node distance=0pt,
 start chain = A going right,
    X/.style = {anchor=base, rectangle, draw,
                minimum width=10ex, minimum height=3ex,
                outer sep=0pt, on chain},
    Y/.style = {rectangle, draw,
                minimum width=10ex, minimum height=3ex,
                outer sep=0pt, on chain, thick},
                        ]
\node[Y] {$z_{w,A}$};
\foreach \i in {{\ldots}1\ldots,$\tilde{s}$,$z_{w,D}$}
    \node[X] {\i};
\draw ( A-1.west) -- node[left=5ex,minimum width=10ex] {5.} ( A-1.west);
\end{tikzpicture}
\\
\begin{tikzpicture}[
node distance=0pt,
 start chain = A going right,
    X/.style = {anchor=base, rectangle, draw,
                minimum width=10ex, minimum height=3ex,
                outer sep=0pt, on chain},
    Y/.style = {rectangle, draw,
                minimum width=10ex, minimum height=3ex,
                outer sep=0pt, on chain, thick},
                        ]
\foreach \i in {$z_{w,A}$}
    \node[X] {\i};
\node[Y] {$z_{w,B}$};
\foreach \i in {$\tilde{s}$,$z_{w,D}$}
    \node[X] {\i};
\draw ( A-1.west) -- node[left=5ex,minimum width=10ex] {6.} ( A-1.west);
\end{tikzpicture}
\\
\begin{tikzpicture}[
node distance=0pt,
 start chain = A going right,
    X/.style = {anchor=base, rectangle, draw,
                minimum width=10ex, minimum height=3ex,
                outer sep=0pt, on chain},
    Y/.style = {rectangle, draw,
                minimum width=10ex, minimum height=3ex,
                outer sep=0pt, on chain, thick},
                        ]
\foreach \i in {$z_{w,A}$,$z_{w,B}$}
    \node[X] {\i};
\node[Y] {$z_{w,C}$};
\foreach \i in {$z_{w,D}$}
    \node[X] {\i};
\draw ( A-1.west) -- node[left=5ex,minimum width=10ex] {7.} ( A-1.west);
\end{tikzpicture}
\end{center}
\caption{
Progression of place-within-level $z$ as the multipath routing algorithm
cycles through the levels of the butterfly network.
\label{fig:route-overview}
}
\begin{center}
\includegraphics{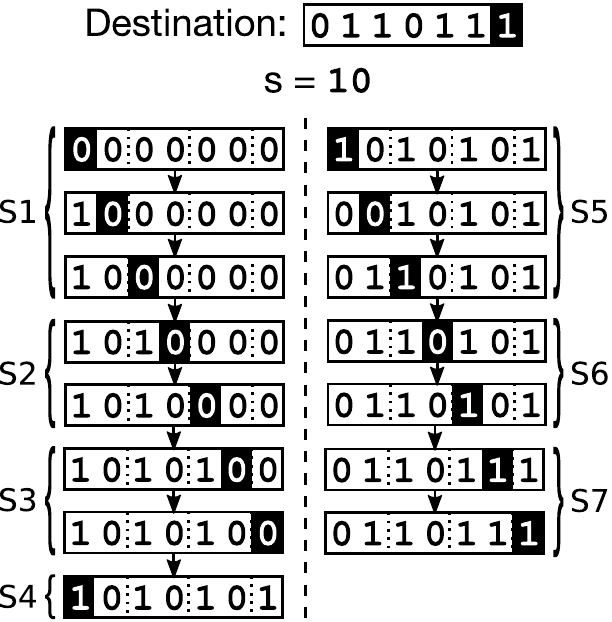}
\end{center}
\caption{
An example of one path as constructed by the proposed multipath
routing algorithm.
The path is shown for $s = 10_2$
and $w = (6, 0110111_2)$.
\label{fig:routing}
}
\end{figure}
}

\subsubsection{Algorithm Specification}

We now begin the formal specification of our multipath routing scheme for the
wrap-around butterfly network.
For convenience, the relevant variables are summarized in Table \ref{tab:routing}.
Utilizing vertex transitivity, we label the source node as
$(l^{(0)}, z^{(0)}) = (0, 0)$ and denote the destination node as $w = (l_w, z_w)$,
without loss of generality.

Let $s$ be an $h$-bit binary string with $s_i$ denoting the bit at index $i$.
There are $2^h$ such strings.
Let $v_s^{(t)} = (l^{(t)}, z^{(t)})$ be the node at position $t$
in the path parameterized by $s$.
For convenience,
we will omit the subscript $s$ when it is obvious from context.
We define three distinct partitions of $m$-bit binary strings.
Let $Q_{v^{(0)}}$ be the set of $m$-bit
strings in which the bits at
all indices $h \leq i < l_w - h$ match those of $z^{(0)}$,
and let $\overline{Q_{v^{(0)}}}$ be its complement.
Note that $Q_{v^{(0)}}$ is trivially all $m$-bit strings if $l_w < 2h$.
Let $R_s$ be the set of $m$-bit strings with the lowest $h$
bits all matching the bits of $s$,
and let $\overline{R_s}$ be its complement.
Let $S_s$ be the set of $m$-bit strings with the $h$ bits
preceding index $l_w$ all matching the bits of $\tilde{s}$,
where $\tilde{s}$ is a cyclic permutation of $s$:
\beq
\tilde{s}_i &=& s_{(i + l_w) \text{ mod } h},
\eeq
and let $\overline{S_s}$ be its complement.
We will make use of the fact that:
\beq
s \neq s^\prime &\implies&
S_s \cap S_{s^\prime} = R_s \cap R_{s^\prime} = \emptyset.
\eeq

Routes are constructed in 7 stages.
The network topology dictates that $l^{(t+1)} = l^{(t)} + 1$ (mod $m$),
so we let $l = t$ (mod $m$).
and that $z^{(t+1)}$ is equal to $z^{(t)}$ with or without the bit in index
$l^{(t)}$ inverted, depending on whether the down or down-right edge was
taken at step $t$.
\begin{description}
\item[Stage 1: ($0 \leq t < h$)]{
Down or down-right edges
are chosen such that the $t$th bit of $z^{(t+1)}$ is equal to the $t$th bit
of $s$.
Throughout Stage 1, all nodes are within the sender's trusted neighborhood.
Throughout Stage 1, $z^{(t)} \in Q_{v^{(0)}}$.
At the end of Stage 1, $z^{(h)} \in S_s$, and $z^{(t)}$ will remain so until the level cycles to $0$ at $t = m$.
}
\item[Stage 2: ($h \leq t < l_w - h$)]{
Edges are chosen to make the $t$th bit of
$z^{(t+1)}$ the inverse of the $t$th bit of $z^{(0)}$.
Note that this stage does not occur when $l_w < 2h$.
If this stage occurs, then $z^{(t)} \in \overline{Q_{v^{(0)}}}$ until these
levels are reached again in stage 6.
}
\item[Stage 3: ($l_w - h \leq t < l_w$)]{
The bits of $z^{(t)}$ are chosen to match $\tilde{s}$,
such that after the stage is complete, $z^{(t)} \in R_s$.
}
\item[Stage 4: ($l_w \leq t < m$)]{
Paths are chosen such that the $t$th bit of $z^{(t+1)}$ matches that of the
destination node $z_w$.
This stage will not occur if $l_w > m - h$.
}
\item[Stage 5: ($m \leq t < m + h$)]{
There are two cases.
If $2h < l_w < m - h$,
then there is no overlap between the indices defining $R_s$ and $S_s$.
In this case, the first $h$ bits of $z^{(t)}$ are set to
match $z_w$.
Otherwise there is some overlap between the indices defining $R_s$ and
$S_s$.
In this case, the each of the first $h$ bits of $z^{(t)}$ is either kept the
same if $l_w - h \leq l < l_w$, or set to the corresponding bit of $z_w$
otherwise.
In this stage and after, $z^{(t)}$ is no longer guaranteed to be in $R_s$.
However, $z^{(t)}$ remains in $S_s$ during and after this stage.
}
\item[Stage 6: ($m + h \leq t < m + l_w - h$)]{
In this stage, edges are chosen to set the bits of $z^{(t)}$ to their
corresponding value in $z_w$.
$z^{(t)} \in \overline{Q_{v^{(0)}}}$ throughout this stage,
but not afterwards.
}
\item[Stage 7: ($m + l_w - h \leq t < m + l_w$)] {
The $h$ bits of $z^{(t)}$
preceding index $l_w$ are set to match $z_w$.
All nodes in this stage are within $h$ hops of $w$ and thus in its trusted
neighborhood.
After this stage, $v^{(m + l_w)} = w$ and routing is complete.
}
\end{description}

\subsubsection{Proof of Path Independence}

\begin{theorem}
Given an $m$-bit wrap-around butterfly network ($m > 1$),
and a radius $h$ ($1 \leq h \leq \left\lfloor \frac{m}{2} \right\rfloor$),
for all node pairs $(v, w)$ such that $d(v,w) \geq 2h$,
there exist $2^h$ paths $v_s$ ($0 \le s < 2^h$) from
$v$ to $w$ such that
$s \neq s^\prime \implies v_s \cap v_{s^\prime} \subset T_h(u) \cup T_h(v)$.
\end{theorem}
\begin{proof}
Nodes from two paths can only coincide if their levels are the same.
Nodes which share a level must either be in the same stage, or 4 stages
apart.
Let ($a$,$a^\prime$) denote a pair of sub-paths corresponding to stage $a$ of
one path and stage $a^\prime$ of another.
Excluding paths that intersect in their trusted neighborhoods, (1,1) and (7,7),
we have reduced the list of possible intersections to the following cases:
(2,2), (3,3), (4,4), (5,5), (6,6), (1,5), (2,6), and (3,7).
Nodes in stages 2--4 belong to $R_s$ so cannot overlap with any stage 2--4
nodes from another path, eliminating (2,2), (3,3), and (4,4).
Similarly, nodes in stages 4--6 belong to a unique $S_s$,
eliminating (5,5) and (6,6).
Nodes in stage 1 belong to $Q_{v^{(0)}}$ while those in stage 5 belong in
its complement, eliminating (1,5).
Similarly, for all $l$ in stage 2, $z^{(l)}$ is equal to $z^{(0)}$,
while in stage 6, $z^{(l)}$ is the inverse, eliminating (2,6)
This leaves only (3,7), a collision which can occur only for only one path
(with $s$ matching the first $h$ bits of $z_w$), and which enters the trusted
neighborhood in stage 3.
For this single path, we can proceed directly from stage 2 to stage 7,
eliminating the last possible collision.
\end{proof}

Thus, assuming the partial trust model with trust transitive
for $h$ hops, we can construct $2^h$ paths on a wrap-around butterfly topology
which do not intersect outside the trusted neighborhoods of the source and
destination.
Note that the node sequence $v_s^{(t)}$ can be calculated entirely
from the source $v$, destination $w$, and path parameter $s$,
meaning that with this information nodes are able to determine which neighbor
to route a given message copy to.
Furthermore, the existence of $2^h$ paths places a lower bound on the
effective redundancy $\delta_h$,
showing that the decentralized, redundant, structured networks such as the
butterfly can have a very low probability of failure when faced with
adversarial faults, even from a very powerful attacker.

\section{Discussion}
\label{sec-discussion}

While decentralized protocols have received much attention for their potential
fault tolerance applications,
centralized structures are always vulnerable to exploitation by non-technical means (i.e., coercion),
and there is a need for a better understanding of the relationship between network
structure and attack-tolerance.
We have proposed a network-based scheme for {\em adversarial} fault tolerance on the
butterfly topology,
utilizing a novel concurrent multipath routing algorithm.
We have also demonstrated how {\em partial trust transitivity},
in addition to being more realistic than infinite transitivity,
provides a theoretical foundation for
quantitative analysis of the relationship between trust,
network structure, and attack tolerance.

Such attacks include many forms of censorship and surveillance,
which have important social implications.
We have already discussed two such cases:
Pakistan's inadvertant censorship of YouTube
\cite{hunter_pakistan_2008}
and the FBI's surveillance-turned-censorship of Lavabit
\cite{poulsen_edward_2013}.
The reader may wonder how our methods could be employed in scenarios
such as large-scale state-sponsored censorship.
Censorship-resistant infrastructure often replaces central servers
(e.g., the router in the 2008 YouTube incident) with multiple servers across
the world, synchronized through consensus protocols.
The {\em directory authorities} used by the Tor project
\cite{dingledine_tor:_2004} are one example.
However, the size of such authority networks is often limited by the number of
trusted relationships (degree) each node can maintain, and the inherent insecurity of
extending transitive trust to an ever-larger network.
Our work fills a much-needed gap by quantifying the connection between network-structure,
trust transitivity,
and attack-tolerance.
We provide both a theoretical framework and specific example of how network structure
can be engineered to leverage trust for a high level of attack-tolerance,
without sacrificing scalability.

Fault-tolerant network infrastructures have many direct applications.
Areas such as cryptocurrency
\cite{mazieres_stellar_2015,nakamoto_bitcoin:_2008,king_ppcoin:_2012},
secure multiparty computation
\cite{yao_protocols_1982,chaum_multiparty_1988,goldreich_how_1987},
and wireless sensor networks
\cite{khiani_comparative_2013}
have immediate need for scalable, fault-tolerant infrastructures.
Many Internet services (e.g., email, social networks, cloud storage)
are still highly centralized and vulnerable to technical and
non-technical (i.e., coercive) attacks.
Fault tolerance using {\em both}
decentralized protocols and
decentralized network structures
is one promising approach to securing these services.

We have focused primarily on adversarial faults that block or
change messages (e.g., censorship).
Existing cryptographic techniques for circumventing surveillance
are relatively mature compared to those for tolerating censorship.
However, the techniques presented in this paper
are entirely compatible with, and in some cases could enhance, existing
anti-surveillance techniques.
For example, {\em man-in-the-middle} attacks exploit a privileged
network position to attack otherwise secure cryptography,
suggesting that structural approaches can complement cryptographic ones.

While our present proposal is specific to the butterfly topology,
the multipath fault tolerance scheme could
be applied to any network that has both sufficient redundancy and a
routing algorithm to discover independent paths.
For general networks, finding all such paths is NP-hard,
but efficient, suboptimal algorithms exist
\cite{reiter_resilient_1998}.
However, we have argued that attack-tolerance requires the ability to
influence network structure and reduce reilance on single points of failure.
Our work is most applicable to cases where the need for attack-tolerance
justifies investment in deliberate infrastructure.
For example, a coalition of groups supporting free expression could use our work
to construct a censorship-resistant communication network.
In general, such groups would need to invest resources into vetting their neighbors
to establish trust, but there are scenarios in which the attack-tolerance
requirement would justify that investment.
It is also worth noting that because faulty paths can be identified in our scheme,
it may at times be appropriate to begin by assuming mutual trust and
revoking that relationship if it is violated.
Any entities dependent on the proper functioning of the network would have an incentive
to resist attack in order to maintain their ability to participate.

Our work suggests several directions for future work towards developing
practical, attack-tolerant communication infrastructure.
The development of new multipath routing algorithms on other structured
networks could achieve higher levels of redundancy.
It is also desirable to identify dynamics that give
rise to structured networks,
and to evaluate whether our results can be generalized to unstructured or
approximately structured networks.
Finally, these results could be implemented to address specific
applications, e.g., secure messaging, domain name resolution, or
anonymous web browsing.

\section{Conclusion}
\label{sec-conclusion}

We have presented a novel concurrent multipath routing algorithm for the butterfly
topology,
as well as a scheme for using this algorithm to construct a highly attack-tolerant
virtual channel between any two nodes, even when no fully-trusted path exists
between them.
Under this scheme, the probability of an adversary causing an undetectable error
decreases exponentially with the network's effective redundancy.
The effective redundancy, in the case of the butterfly topology,
grows exponentially with the trust radius.
Furthermore, a small increase in the number of messages sent can compensate
for a large increase in the number of messages compromised by an adversary.
We have also demonstrated how the assumption of partial trust
transitivity can enable a quantitative analysis of the
relationships between network structure, trust, and attack-tolerance.
These results are directly applicable to systems in which the link
structure can be imposed by the designer,
and more generally, provide a theoretical foudnation that can be used more to evaluate
the role of network structure, trust transitivity, and effective redundancy
on attack-tolerance.

\section{Acknowledgments}
The authors would like to thank Tony Garnock-Jones, A. Frederick Dudley, and
Nathaniel Bezanson for helpful conversations.
This research was partly supported by the National Science Foundation under Grant No. IIS-1617820.

\bibliographystyle{abbrv}
\bibliography{preprint} 
\end{document}